\documentclass[a4paper,UKenglish]{lipics-v2018}

\usepackage{microtype}

\newtheorem{lem}{Lemma}
\newtheorem{theo}{Theorem}
\newtheorem{coro}{Corollary}

\theoremstyle{definition}
\newtheorem{defi}{Definition}

\title{An EPTAS for machine scheduling with bag-constraints.}

\titlerunning{An EPTAS for machine scheduling with bag-constraints}

\author{Kilian Grage}{Department of computer science, Kiel University, 24118 Kiel, Germany}{}{}{}

\author{Klaus Jansen}{Department of computer science, Kiel University, 24118 Kiel, Germany}{}{}{}

\author{Kim Manuel Klein}{EPFL, 1015 Lausanne, Switzerland}{}{}{}

\authorrunning{K. Grage et al.}

\Copyright{Kilian Grage, Klaus Jansen, Kim Manuel Klein}

\subjclass{}


\category{}

\funding{Research supported by German Research Foundation (DFG) JA 612 /20-1,}

\EventEditors{John Q. Open and Joan R. Access}
\EventNoEds{2}
\EventLongTitle{36th International Symposium on Theoretical Aspects of Computer Science (STACS'19)}
\EventShortTitle{STACS 2019}
\EventAcronym{STACS}
\EventYear{2019}
\EventDate{March 13-16, 2019}
\EventLocation{Berlin, Germany}
\EventLogo{}
\SeriesVolume{42}
\ArticleNo{XX}
\nolinenumbers 
\hideLIPIcs  

\begin{document}

\maketitle

\begin{abstract}
Machine scheduling is a fundamental optimization problem in computer science. The task of scheduling a set of jobs on a given number of machines and minimizing the makespan is well studied and among other results, we know that EPTAS's for machine scheduling on identical machines exist. Das and Wiese initiated the research on a generalization of makespan minimization, that includes so called bag-constraints. In this variation of machine scheduling the given set of jobs is partitioned into subsets, so called bags. Given this partition a schedule is only considered feasible when on any machine there is at most one job from each bag. 

Das and Wiese showed that this variant of machine scheduling admits a PTAS. We will improve on this result by giving the first EPTAS for the machine scheduling problem with bag-constraints. We achieve this result by using new insights on this problem and restrictions given by the bag-constraints. We show that, to gain an approximate solution, we can relax the bag-constraints and ignore some of the restrictions. Our EPTAS uses a new instance transformation that will allow us to schedule large and small jobs independently of each other for a majority of bags. We also show that it is sufficient to respect the bag-constraint only among a constant number of bags, when scheduling large jobs. With these observations our algorithm will allow for some conflicts when computing a schedule and we show how to repair the schedule in polynomial-time by swapping certain jobs around. 
\end{abstract}

\newpage
\section{Introduction}\label{S:Introduction}

The machine scheduling problem is a classical optimization problem known in computer science. It stems from the simple idea of having multiple jobs that need to be scheduled on a set of machines. Formally we define this problem, also known as \textit{makespan minimization}, as follows: Given a set of jobs $J$ such that each job $j \in J$ has a height or processing time denoted with $p_j$ and $m$ machines. The task is to find a schedule that assigns all jobs to machines and minimizes the makespan, which is the biggest load on any machine. This problem is known to be strongly NP-hard and therefore approximation algorithms are being studied for makespan minimization and different variations of this problem.

Especially polynomial-time approximation schemes (in short PTAS) have been studied for makespan minimization. A PTAS for a minimization problem is a family of algorithms $(A_\epsilon)_{\epsilon>0}$ such that for every fixed $\epsilon >0$ and for every instance $I$ with optimal value $OPT_I$ the algorithm $A_\epsilon$ yields a solution of at most value $(1+\epsilon) OPT_I$ in a running time polynomial in the size of the input $|I|$. We call $(A_\epsilon)_{\epsilon>0}$ efficient PTAS (EPTAS) if it has a running time of the form $f(\frac{1}{\epsilon})* |I|^c$ for some $c\in O(1)$ and a not necessarily polynomial function $f$. Finally if the running time is fully polynomial in $|I|$ and $\frac{1}{\epsilon}$ we call $(A_\epsilon)_\epsilon$ fully PTAS (FPTAS).

For the classical makespan minimization problem it is known that the problem is strongly NP-hard and there are known PTAS's \cite{schedptas,schedptasimpro} and furthermore EPTAS's \cite{jeptas1,jeptas2}. One known generalization of the machine scheduling problem is to change the model of machines. In the unrelated machines model, the height or processing time of jobs also depends on the machine it is running on. For this variant there are known $2$-approximations by Lenstra, Shmoys and Tardos \cite{unrelated1} and an improvement with a $2-\frac{1}{m}$- approximation by Shchepin and Vakhania \cite{unrelated2}. Also a lower bound of $\frac{3}{2}$ is known for unrelated machines \cite{unrelated1}.

We will only look at identical machines in this paper and we want to consider a variant of machine scheduling that involves conflicts. In this variation some jobs are in conflict with each other and therefore not allowed to be executed on the same machine. An easy and intuitive way to model conflicts is to use a graph, where each job is represented as a node and an edge between nodes indicates a conflict. The problem of minimizing the makespan under these restrictions given an arbitrary conflict-graph is NP-hard. A tight $2$-approximation is known for the case where the conflict graph is polynomial-time colorable \cite{conflict}.

A special case of this problem is given when the conflict graph is a cluster graph, that consists of multiple components such that each component is a clique. In this case we can easily model the given conflicts as sets of jobs and each set contains all nodes of one clique. The problem of scheduling jobs under these type of conflicts is also known as machine scheduling with bag-constraints \cite{bagptas}, whereas the mentioned sets are called bags.

\subsection{Machine scheduling with bag-constraints}

In software engineering, especially for parallel and distributed systems, it is not unusual that one wants to enforce different tasks to run on different machines/processors. This can be done to schedule tasks more efficiently on parallel machines but also for system stability and security purposes. To prevent failure and crashes, jobs need to be run separately, so in case that one machine fails, other machines can still continue working. 

To model this, one can extend the definition of machine scheduling by partitioning the set of jobs in subsets $B_1,B_2,...,B_b$. We expect that each job $j\in J$ is contained in exactly one of these subsets, that we call \textit{bag}. For the solution we expect again a schedule with minimum makespan, but additionally we only allow that at most one job from each bag is scheduled on a single machine. Further we will call a violation of the bag-constraints, given by two jobs of the same bag on one machine, a \textit{conflict} and say the respective jobs are conflicting each other. In this paper we will consider only identical machines and our goal will be to minimize the given makespan. The resulting problem is known as \textit{machine scheduling with bag-constraints}.

This problem was recently studied by Das and Wiese and they developed a PTAS for machine scheduling with bag-constraints on identical machines. They also considered the case of unrelated machines and showed a lower bound for the approximation ratio and giving an 8-approximation for the case where all jobs of each bag can go on the same machines \cite{bagptas}. The known PTAS uses a dynamic program to schedule large jobs like in an optimal solution. Given this initial distribution they use different techniques, like flow-networks and greedy algorithms to build a schedule of bounded height. Since the problem is strongly NP-hard we cannot expect to find an FPTAS unless $P=NP$. In this paper we will close this gap between known PTAS and non-existing FPTAS and solve one of the main open problems of Das and Wiese by giving the first EPTAS for machine scheduling with bag-constraints.

The main difficulty for giving an EPTAS for this problem is, that the known strategies for EPTAS from other scheduling problems cannot be applied for this problem without further modifications. The common strategy for EPTAS' for makespan minimization is to find an efficient placement of large jobs, usually with a (mixed) integer linear program (in short (M)ILP) and place small jobs with greedy algorithms \cite{jeptas1,jeptas2}. In the presence of bag-constraints the problem arises that not every efficient placement of large jobs allows for a schedule of bounded height, as seen in figure \ref{abb:diffic}. The first idea that comes into mind is to incorporate the bag-constraints into the MILP. To do this however we would need a number of integer variables depending on the number of bags. As we can have as many bags as jobs we would need a running time exponential in the size of the instance $|I|$ to solve such an MILP.

\begin{figure}
 \centering
  \includegraphics[scale = 0.35]{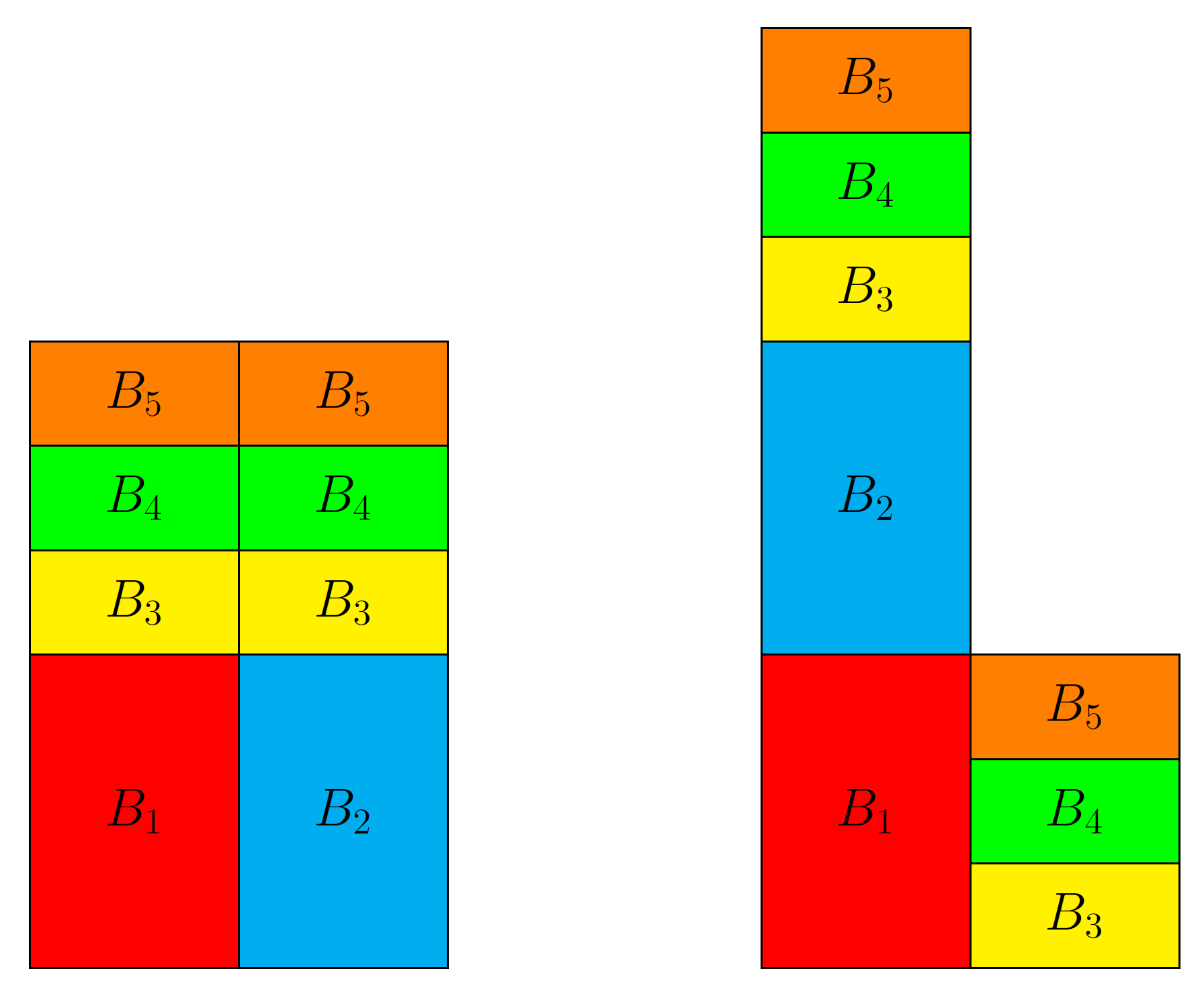}
	 \caption{}
	{\footnotesize Comparing possible schedules depending on large job placement. In the left schedule the two large jobs were scheduled on different machines, allowing for a schedule of height OPT. Considering the placement of large jobs in the right schedule, we see even though large jobs are packed with height OPT, we are forced to place small jobs like depicted and increase the overall makespan.}	
  \label{abb:diffic}
\end{figure}

We overcome these problems by relaxing the bag-constraints to some extent. While computing a schedule we will temporarily allow for some conflicts of jobs from the same bag, but in a controlled way, such that we can repair our schedule in polynomial-time and with a small increase of the overall makespan. We do this by splitting bags into two groups $G_1$ and $G_2$ with $G_1$ only holding a constant number of bags. We then show that to find a feasible distribution of all large jobs, that it is sufficient to schedule all large jobs from $G_1$ without conflicts. Even when multiple jobs of the same bag in $G_2$ are on one machine, we can repair the schedule without increasing the makespan. We use this insight to distribute large jobs to machines with an MILP using a constant number of integral variables.

For the second group $G_2$ especially the placement of the small jobs can cause problems, since  we can have many bags with a large number of small jobs. Additionally these jobs might conflict with large jobs already in place. To solve this problem we introduce a new instance transformation that we apply before constructing and solving our MILP. This transformation splits large and small jobs of bags in $G_2$ in separate bags, while only increasing the height of the schedule marginally. To schedule bags only containing small jobs we can apply some easy heuristics such as largest processing time first (LPT) \cite{graham} with some extensions to respect bag-constraints. Given a solution of the transformed instance we can easily construct a solution of the original one in polynomial-time, yielding us an EPTAS.
 
\begin{theo}
There exists an EPTAS for the machine scheduling with bag-constraints on identical machines.
\end{theo} 

In the following we will first introduce some preprocessing steps, classification of bag groups and the aforementioned instance modification as well as its correctness (section 2). Knowing that we can work with a modified instance of easier structure we then formalize the MILP and show how we can repair a given MILP solution in case the solution does not allow for a feasible distribution of jobs (section 3). We finalize our presentation of the EPTAS by showing how small jobs can be scheduled such that the height of the schedule does not increase arbitrarily in the end (section 5). Due to space limitations we moved proofs into the appendix.


\section{Preliminaries}\label{S:Prel}

In the following we will assume that we are given an instance $I$ of the machine scheduling with bag-constraints problem with a set of jobs $J$ that is separated in a partition of bags $B_1 ,..., B_b$ and a number of machines $m$. Furthermore an approximation ratio is given by $\epsilon$ and we will assume without loss of generality that $\frac{1}{\epsilon}$ is integral. For this arbitrary instance we compute a solution that has a makespan of at most $(1 + O(\epsilon)) OPT_I$. It is sufficient to find an $(1+O(\epsilon)) OPT_I$ schedule as we can replace our input $\epsilon$ with $\epsilon' = \frac{1}{c} \epsilon$ for $c \in O(1), c \neq 0$ to gain a $(1+\epsilon) OPT_I$ schedule. In the following we will also write $log$ for $log_2$ unless another base is specified.

\subsection{Classification of jobs and bags}

Before we start classifying jobs and bags, we will apply some standard scaling and rounding techniques to reduce the number of item sizes. With a binary search frame work we may assume that we know the height of an optimal makespan $OPT$ and by scaling we may assume that $OPT = 1$. Further we will round up all job lengths to the next power of $1+ \epsilon$. With this we may assume that for any $j\in J$ we have that $p_j =(1+\epsilon)^k$ for some $k \in \mathbb{N}$. Note that this rounding increases our optimum to $1+\epsilon$ \cite{bagptas}. 

We later want to use a result from Das and Wiese to schedule medium size jobs. For that reason we will use the same classification of jobs that is based on the following lemma. As Das and Wiese also gave a very short and nice proof by contradiction, we will omit the proof in this paper and refer to \cite{bagptas}.

\begin{lem}\label{lem:medk}{}
We can compute $k \in \mathbb{N}_{\le \frac{1}{\epsilon^2}}$ such that
$\sum\limits_{j\in J : p_j \in [\epsilon^{k+1},\epsilon^{k})}{p_j} \le \epsilon^2 *m$.
\end{lem}

For the rest of this paper we will assume that $k$ is set to be the parameter of this lemma. We will further classify jobs just like in \cite{bagptas} with this $k$ as follows: We call a job $j$ \textit{large} if $p_j \ge \epsilon^k$, \textit{medium} if  $\epsilon^k > p_j \ge \epsilon^{k+1} $ and \textit{small} if $p_j < \epsilon^{k+1}$. Finally we also want to respect bags that have a large amount of medium and large jobs. We call a bag $B_l$ \textit{large} when $B_l$ holds at least $\epsilon *m$ jobs that are medium or large. If a bag is not large we call it \textit{small}. We can also note that the amount of large bags is bounded by $O(\frac{1}{\epsilon^{k+2}})$ \cite{bagptas}.

For our EPTAS we will need another type of bag-classification. To be able to schedule large and medium jobs efficiently we need to differentiate between bags that we are going to prioritize and bags that are less important. The property we want to enforce for such an important bag is that they hold a large amount of jobs with a certain large size. Therefore we introduce a notation for so called size-restricted bags. Further we define a functions to represent a sorted list of size-restricted bags of one size.

\begin{defi}
Let $B_l$ be any bag and $s$ any item size. With $B_l^s:= \{j\in B_l | p_j = s\}$ we denote the set of all jobs in $B_l$ with size $s$ and call $B_l^s$ \textit{size-restricted} bag with size $s$. Also we define for each large item size $s$ a function $o_s:\mathbb{N}_{\le b} \Rightarrow \mathbb{N}_{\le b}$ to be a bijective index function such that for every $l< b$ we have that $|B^s_{o_s(l)}| \ge  |B^s_{o_s(l+1)}|$.
\end{defi}

Given a size-restricted bag $B_l^s$ we may also refer to $B_l$ as the respective full bag. In order to find a feasible distribution of all large jobs we need to ensure that a constant number $b' \le b$ bags of each large item size is packed with no violations of the bag-constraints. Even if the rest of the bags are placed such that a machine holds multiple jobs of one bag, we can repair the schedule to gain an overall feasible schedule. The constant $b'$ depends on the number of medium jobs any machine can hold in an optimal solution. We therefore set $q := \frac{1+2*\epsilon+\epsilon^2}{\epsilon^{k+1}}$ to be this number of jobs. Note that this $q$ also results from a modification that we have yet to introduce, which will increase the height of an optimal solution up to $T:= 1+2\epsilon +\epsilon^2$. We formally define this set of bags as follows:

\begin{defi}\label{def:prio}
  Set $b':= (d*q+1)*q$ with $q$ being the number of medium jobs any machine can hold in an optimal schedule as introduced before the lemma and $d\in O(log_{1+\epsilon}{\frac{1}{\epsilon^{k}}})$ being the number of item sizes of large jobs. We call a bag $B_l$ \textit{priority bag} if and only if there exists a large item size $s$ and a position $i\le b'$ such that $o_s(i) = l$.  Further we define that every large bag is a priority bag as well. If a bag $B_l$ by this definition is not a priority bag, we call $B_l$ \textit{non-priority bag}.

\end{defi}

For the rest of the paper we set $b'$ just like in this definition. Intuitively we look at all size restricted bags (for large item sizes) and take the first $b'$ bags in the sorted lists respectively and call these and their respective full bags priority bags.
Note that we can add large bags to the set of priority bags as their number is also bounded by $O(\frac{1}{\epsilon^{k+2}})$. This gives us the advantage that all non-priority bags are small bags. We will further modify small bags to ensure we can place small jobs and large jobs from non-priority bags independent from each other.

\subsection{Instance Transformation} 

The first step of our algorithm is to make further modifications on the instance. In order to make placing non-priority bags easier, we will split large and small jobs in separate bags. We further add small jobs for every medium and large job, so that after finding a solution for this modified instance, we also can revert the modifications to get a solution for the instance with the original bags.

 Consider therefore the following process: Let $B_l$ be any non-priority bag and $p_{max}$ the height of the highest job that is still only a small job in $B_l$ (if $B_l$ should hold no small jobs, then we do not modify $B_l$ and continue with the next bag). Open up a new bag $B'_l$ that contains all large Jobs of $B_l$. And for all large and medium jobs $j \in B_l$ replace $j$ in $B_l$ with $\overline j$ such that $p_{\overline j} = p_{max}$. We will call these additionally added small jobs \textit{filler-jobs}.  Alter every non-priority bag in this way and we will end up with a modified instance that we denote with $I'$.
 
 Intuitively we shortened all medium and large jobs to small jobs in non-priority bags, while saving copies of the large jobs in separate bags, as depicted in figure \ref{abb:modtra}. Note that we removed medium jobs from non-priority bags. We will show how to use a result of Das and Wiese \cite{bagptas} to add these back to our instance. Before doing so we will first show that by applying this modification we overall lose only an $\epsilon$ factor in the objective.
 
 \begin{figure}
  \centering
	\noindent
  \includegraphics[scale = 0.4]{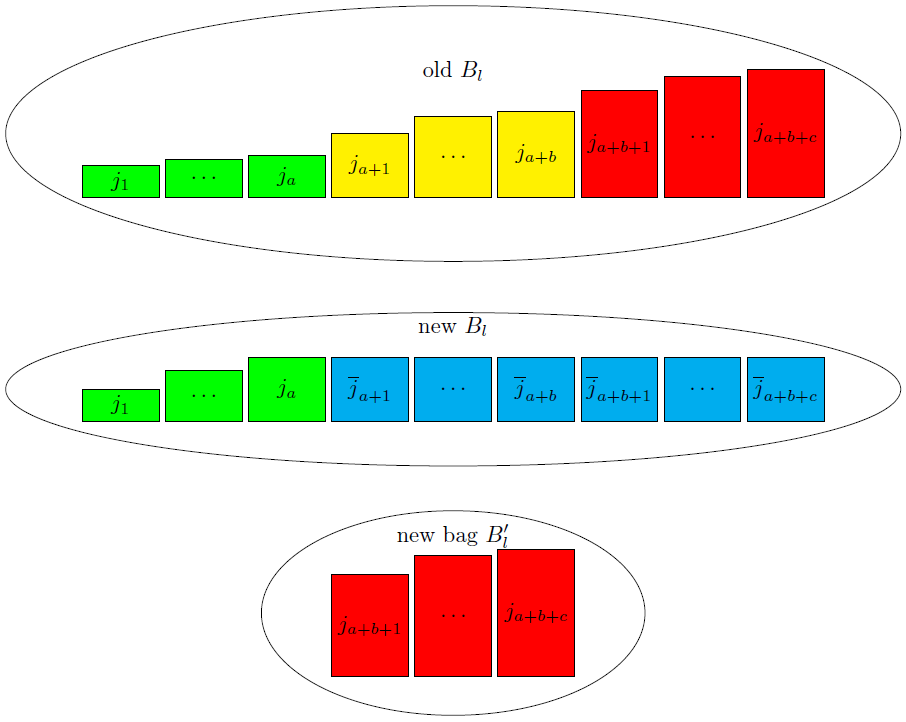}
	 \caption{}
	{\footnotesize Example of a transformation of non-priority bag $B_l$. Jobs are coloured, large jobs are red, medium jobs yellow and small jobs green. Additionally filler-jobs are coloured blue to differ from normal small jobs. As seen in the picture each bag will be separated in two bags. One containing all large jobs and another one containing only small jobs. }	
  \label{abb:modtra}
\end{figure}

\begin{lem}\label{lem:mod1}
Assume we are given an instances for machine-scheduling with bag-constraints $I$ and let $I'$ be the modified instance of $I$. If there is a solution for $I$ with makespan $C$ then there is also a solution for $I'$ with makespan $(1+\epsilon)*C$.
\end{lem}

The advantage of this modification is that we only have non-priority bags that contain only large jobs or only small jobs. Thanks to this it is sufficient to schedule large and small jobs independently for these bags as they cannot conflict with each other. We can also see that we can apply this transformation in polynomial time and since we increase our number of jobs by at most factor $2$, this will not affect the overall running time of our algorithm except for this constant factor.

Now we will show how we can construct a solution with bounded height for the original instance given a solution for the modified one. First we will deal with the task to find a valid placement of medium jobs from non-priority bags as these were completely discarded in our modified instance. We will add these medium jobs to the extra added bags with large jobs and add them to the schedule such that there is no conflict among any of these jobs. To do so we will use a result of \cite{bagptas}. As it might not be obvious that we can apply the lemma of \cite{bagptas} we will give a rephrase of this lemma and also a proof. 

\begin{lem}\label{lem:mod2}
Let $I$ be an instance of the machine scheduling with bag-constraint problem and let $I'$ be the modified instance. Let $S'$ be a solution of the instance $I'$. We can expand $S'$ by adding and scheduling all medium jobs from non-priority bags of instance $I$ such that no medium and large job from the same bag $B_l$ in $I$ are on the same machine in $S'$. Additionally this schedule will increase the makespan of the solution $S'$ by at most $2*\epsilon \in O(\epsilon)$.
\end{lem}

With this lemma we gain a solution $S'$ for Instance $I'$ and this solution also contains all jobs from instance $I$. The only difference is that $S'$ may have put a small job and either a medium or a large job of the same bag in $I$ on the same machine, as these jobs were separated in two bags in $I'$. Using filler-jobs from $I'$ we can fix these conflicts and therefore generate a solution for $I$.

\begin{lem}\label{lem:mod3}
Let $I$ be an instance of the machine scheduling with bag-constraint problem and let $I'$ be the modified instance. Let $S'$ be a solution of the instance $I'$ that additionally contains all medium jobs from non-priority bags of instance $I$ as per lemma \ref{lem:mod2}. Then we can find a valid solution $S$ for instance $I$ that has at most the same makespan as $I'$.
\end{lem}

With the last two lemmas we conclude that our modification not only gives us an easier structure to work with, but also keeps the overall error small. For the following we therefore assume that we continue working with our modified instance $I'$ and will denote this with $I$ in the next sections. Due to the modification and lemma \ref{lem:mod1} we know that our estimated optimal value for this instance increases up to $T = 1+2\epsilon+\epsilon^2$ (remember that our previous optimum was $1+\epsilon$ due to rounding).

For the rest of this paper we will show how to compute a $1+O(\epsilon)$ for our modified instance. With the last two lemmas \ref{lem:mod2} and \ref{lem:mod3} we know we can generate a schedule for our original instance in polynomial by only increasing the makespan by at most $2\epsilon$. This will overall lead to a $1+O(\epsilon)$ schedule. In the next section we will introduce our MILP that will allow us to find this schedule.


\section{MILP}\label{S:MILP}

The main idea for the MILP is to build a modified configuration-LP. The configurations or patterns will only contain large and medium jobs. We will further focus on placing large and medium jobs from priority bags without conflicts, while we only use place-holder slots for non-priority bags and allow multiple jobs of non-priority bags on a machine. We will also consider positions for small jobs to ensure that we can place small jobs and account for possible conflicts with large and medium jobs from priority bags. Therefore we will also consider fractional packing of smaller jobs on top of patterns, while we enforce some small jobs from priority bags to be scheduled with integral variables. First of all we give a formal definition for these patterns, which partially resembles the pattern definition of Das and Wiese \cite{bagptas}. In our version however we do not need entries for every bags, but only for priority bags. For non-priority bags we introduce $B_x$ and $B_x^s$ for all item sizes $s$ to hold all jobs from non-priority bags (with respective size $s$), giving us overall a smaller amount of possible patterns.

\begin{defi}\label{def:pat}
Let $p$ be a tuple with up to $q = \frac{1+2*\epsilon+\epsilon^2}{\epsilon^{k+1}} $ entries of the form $B_l^s$ with $s$ being any medium or large item size and $B_l$ either being a priority bag with $l\le b'$ or $B_l =B_x$ for indicating a slot for an arbitrary non-priority bag.

We call $p$ a valid pattern if and only if: $\sum\limits_{B_l^s \in p}{s} \le T = 1+2\epsilon+\epsilon^2$ and for every priority bag $B_l$ there is at most one entry of a size-restricted bag of $B_l$ in $p$. A valid pattern may hold an arbitrary amount of $B_x$ entries. Furthermore define $\mathscr{P}$ to be the set of all valid patterns.
\end{defi}

One can imagine one entry $B_l^s$ of a pattern $p$ to be a slot reserved to the bag $B_l$ with size $s$. In case of an entry $B_x^s$ this slot will be reserved for a job of size $s$ of any non-priority bag. We want to remark that the relevant information of these entries are the bags and the size of the slot, we merely use $B_l^s$ as a notation for these entries instead of the pair $(l,s)$.

In our MILP we use integer variables $x_p$ for all patterns $p\in \mathscr{P}$. $x_p$ indicates how many machines will hold patterns of the form $p$. To consider a placement of small jobs for later we introduce variables of the form $y_p^{B_l^s}$ for any bag $B_l$ and any small item size $s$. These variables indicate how many jobs of $B_l^s$ we are going to place on top of a pattern $p$. A majority of these variables are going to be fractional but a constant number of $y$ variables will be integral, to ensure we can pack these jobs later without increasing the makespan by too much. In the MILP we will also make sure that the number of jobs and the assigned area to a pattern may not grow too large.

To formally define the MILP we are going to introduce some additional notations. First of all let $A$ be set of all indices of priority bags such that $l\in A$ means that $B_l$ is a priority bag and set $S_{small}$ to be the set of all small item sizes and $S_{ml}$ to be the set of all medium and large item sizes.
With $B_x^s := \bigcup\limits_{l\le b:l\not\in A} B_l^s$ , we denote the set of all jobs of non-priority bags of a certain size $s$. With $height(p)$ we will denote the height of the pattern $p$ which will be the sum off all included jobs in the pattern. We also remember that we set $T=1+2\epsilon+\epsilon^2$ to be the optimal height of our given instance.
Furthermore we define a characteristic function $\chi_p$ for patterns $p\in \mathscr{P}$ to indicate how many jobs of a bag appear in a pattern as follows: Given a potential entry $B_l^s$ of $p$ with $s\in S_{ml}$ and $l \in A \cup \{x\}$ we define
\begin{equation*}
\chi_p(B_l^s) :=  
\begin{cases}
       z &\quad\text{if } B_l^s \in p \text{ and } B_l^s \text{ appears } z \text{ times in  } p\\
       0 &\quad\text{otherwise.} \\ 
\end{cases}
\end{equation*}
 
 We extend this definition for full bags, let $B_l$ be any bag and set 
\begin{equation*}
\chi_p(B_l) :=  
\begin{cases}
       1 &\quad\text{if } \exists s\in S_{ml}:\chi_p(B_l^s) \neq 0 \\
       0 &\quad\text{otherwise.} \\ 
\end{cases}
\end{equation*}
Note that in this definition  $\chi_p(B_l)=0$ in case of $B_l$ being a non-priority bag, as the $ B_x$ occurrences count for no original bag of the instance.
With these notations we can now construct our MILP.

\begin{alignat}{3}
 && \sum\limits_{p \in \mathscr{P}}{x_p} \le m \quad & \\
&&\sum\limits_{p \in \mathscr{P}}{(x_p *\chi_p(B_l^{s}))} \ge |B_l^{s}| \quad&  \forall s \in S_{ml}, l \in A \cup \{x\}\\
&&\sum\limits_{p \in \mathscr{P}}{y^{B_l^{s}}_p} \ge |B_l^{s}| \quad&  \forall s \in S_{small}, l \le b\\
&&\sum\limits_{s \in S_{small}, l \le b}{y^{B_l^{s}}_p} *s \le x_p * (T- height(p)) \quad&  \forall p \in \mathscr{P}\\
&&\sum\limits_{s_i \in S_{small}}{y^{B_l^{s_i}}_p}  \le x_p * (1 - \chi_p(B_l)) \quad&  \forall p \in \mathscr{P}, l \le b\\
&& x_p \in \mathbb{N}_{\ge 0} \quad&  \forall p \in \mathscr{P} \\
&& y_p^{B_l^s} \in \mathbb{N}_{\ge 0} \quad&  \forall p \in \mathscr{P}, s\in S_{small},  s> \epsilon^{2k+11}, l \in A \\
&& y_p^{B_l^s} \in \mathbb{R}_{\ge 0} \quad&  \forall p \in \mathscr{P}, s\in S_{small},  s \le  \epsilon^{2k+11}, l\in A \\
&& y_p^{B_l^s} \in \mathbb{R}_{\ge 0} \quad&  \forall p \in \mathscr{P}, s\in S_{small},  l\le b, l \not \in A
\end{alignat}
The first condition of our MILP will ensure that we consider at most $m$ patterns as we only have $m$ machines to fill. Conditions (2) and (3) take care that every job of our instance will be placed on some machine. Constraint (2) will specifically check for an appearance of all large and medium jobs in patterns while (3) will ensure every small job will be placed somewhere. The constraint (4) will ensure that the average area that is scheduled on top of a pattern does not exceed the optimal height $T$. To respect conflicts among priority bags we added constraint (5). Constraint (5) first will not allow to pack any small job of a bag $B_l$ on top of a pattern $p$ that already holds large or medium jobs of $B_l$. Additionally this constraint will also only allow to place at most as many jobs of one bag $B_l$ on top of a pattern $p$ as the amount of machines the pattern was assigned to. This will help us later on when placing small jobs and will prevent conflicts between small jobs itself.

The variable constraints ensure that all $x$ variables are integral and most of the $y$ variables are fractional. The only exception are $y_p^{B_l^s}$ variables that belong to a priority bag $B_l$ with $l\in A$ and have an item size $s> \epsilon^{2k+11}$. We make sure that items of this size are packed as full jobs, as we are later going to round the fractional small jobs of priority bags and we need to ensure that the height increase through this rounding is bound by $O(\epsilon)$. We will still see that the number of integral variables is bound by constants dependent on $\epsilon$.

Before we look at the running time necessary to solve this MILP and the number of variables we want to proof, that this MILP will yield a solution, when there exists a solution for our instance.

\begin{lem}\label{lem:milp}
Given a modified instance $I$ of machine scheduling with bag-constraints. If there is a solution of makespan $T$ for $I$ then the MILP will have a valid solution.
\end{lem}

Now that we know that our MILP will have a valid solution the question remains how fast we can compute this solution. We want to apply the result of Kannan \cite{kannan}, which is an improvement on the result of Lenstra \cite{lenstra}. 
The time to solve the MILP with these approaches depends strongly on the number of integral variables. Therefore we will look at this number in the following lemma.

\begin{lem}\label{lem:runtime}
Given an instance $I$ for machine scheduling with bag-constraints. We can solve the respective MILP with a running time of $2^{O(z*log(z))}*poly(n)$ with \\$z = 2^{O(\frac{1}{\epsilon^{k+1}}*log(\frac{1}{\epsilon^{2k+11}}*log^3(\frac{1}{\epsilon})))}$.
\end{lem}

With this lemma we also see that we achieved our desired running time that we may write as $f(\frac{1}{\epsilon}) * poly(n)$ for the respective function $f$. For the rest of the paper will show how to generate a schedule based on a given MILP solution. In the following we may assume that instead of slots we have already distributed specific jobs to slots, when the bag and job-size is already specified. This assumption can be made as we assign large and medium jobs of priority bags to each pattern. For small jobs we will assume that we can do the same and that this packing will prioritize jobs that can be distributed fully on one machine. One can achieve this type of job to slot distribution by scheduling full jobs first, so if there is a $y_p^{B_l^s}$ variable with $y_p^{B_l^s} \ge i$ for some $i \in \mathbb{N}$ then distribute $i$ jobs from $B_l^s$ on this machine and do the same for all variables of the same bag. After doing so we can arbitrarily fill up the fractional parts with any combination of jobs that are left. We will sort out a feasible distribution of small jobs to machines later. In the next part we will find a distribution of all large and medium jobs.

\subsection{Large jobs and medium jobs}

We will start building a schedule by distributing large and medium jobs. We can note that due to our modified instance that all medium jobs are contained in priority bags and all large and medium jobs from these priority bags have a definitive slot in patterns assigned to them. For non-priority bags we only assigned slots of certain sizes but no specific bag. This may lead to potential conflicts among large jobs of non-priority bags, if we try to pack everything like in the MILP solution. In fact we can repair these conflicts by swapping conflicting jobs with already well placed jobs from priority bags. By doing so we gain a feasible schedule without increasing the makespan of the MILP solution. 

\begin{lem}\label{lem:large}
Given an MILP solution for a modified instance of the machine scheduling with bag-constraints, we can find a placement of all large and medium jobs in polynomial time, such that no two jobs of the same bag are placed on the same machine and the load of medium and large jobs on each machine is the same load as assigned in the MILP.
\end{lem}

With this we conclude that, given an MILP solution, we are able to schedule all large and medium jobs. As we can also see, this swapping strategy used in this proof, runs in polynomial time as the swapping just needs to check each machine for a fitting replacement of each group and for potentially all large jobs from non-priority bags. Note that due to the swapping argument we may have changed up some patterns. When placing small jobs from priority bags, we will work with the patterns initially used by the MILP solution, which may lead to conflicts later on. We will show how to solve these conflicts with a similar swapping argument, when we have distributed small jobs in the next section.


\section{Small jobs}\label{S:Small}

The final thing to do is to distribute small jobs to machines. So far we are given a fractional distribution of theses jobs via our MILP solution and the respective $y$ variables. We schedule small jobs in two steps. In the first step we will assign jobs from bags to groups of machines that have similar height such that no job on any machine of the group will cause a conflict. Knowing these groups we then use a greedy LPT-based approach to schedule jobs to specific machines in that group. More precisely we distribute jobs to these groups using a generalization of LPT that respects bag-constraints. We call the following algorithm bag-LPT: 

Assume we have $m'$ machines and bags $B_1,..,B_r$ and for simplicity each bag holds exactly $m'$ jobs, that can go on any machine without conflict. We may fill up bags with dummy-jobs of height $0$ if necessary. For all $i\le r$ sort all jobs of $B_i$ decreasing by height and all machines increasing by load. Now we schedule the $j$-th job on the $j$-th machine given by the respective sorted lists.

Note that this algorithm differs slightly from the greedy algorithm introduced in \cite{bagptas}, since we require that all machines are free for all bags. As we cannot guarantee this for all $m$ machines, we will use this algorithm on different groups of machines and we will ensure that all assigned jobs can run on any machine in the designated group. With that assignment we will use bag-LPT restricted to this group of machines and the set of jobs assigned to that group. With this prerequisite and some properties of LPT, we can conclude with the properties of LPT that this algorithm is good at distributing the overall area of assigned jobs, when the respective machines of a group have similar load.

\begin{lem}\label{lem:baglpt}
Given bags $B_1,..,B_r$ with at most $m'$ jobs each and $m'$ machines that all have the same height $h$. Let $p_{max}$ be the maximum job height of any given bag. Bag-LPT will schedule all jobs such that in the resulting schedule any two machines differ in height by at most $p_{max}$.
Further let $A$ be the total summed up area of all jobs in $B_1,..,B_r$ and write $A=m'*x$ for some $x\in \mathbb{R}$. Then the highest machine of the resulting schedule has a load of at most $h+x+p_{max}$.
\end{lem}

We remark that in case that machines are not on the same height, the loads of machines grow closer to each other as smaller machines obtain larger jobs and vice versa. This goes on until one machine overtakes the other. With this observation our aim for placing small jobs is to build groups of machines of similar height and assign them a set of jobs from bags. By restricting the total assigned area we can use lemma \ref{lem:baglpt} to show that the height of machines does not increase by too much.
The process of grouping machines and assigning them jobs differs for priority and non-priority bags. We discuss this process in the following for both bag-types separately, starting with non-priority bags.

\subsection{Non-priority bags}

Before scheduling non-priority bags we first make an assumption on the required space for priority bags. For every pattern used in the MILP we will assume that the assigned area of small jobs from priority bags is evenly distributed on all machines that hold the respective pattern.
 So for a given pattern $p$ that was assigned to $x_p >0$ machines, we will assume that the load on every of these machines is $\frac{A_p}{x_p}$ with $A_p$ being the total area of fractional small priority-bag jobs assigned to $p$. With this assumption we show that our final schedule after placing non-priority bags will have a height of $1+O(\epsilon)$. We later show that we can pack priority bags not exactly as in this assumption, but the resulting error will be bound by $O(\epsilon)$ as well.

Now to schedule non-priority bags we first round up the heights of every machine to the next multiple of $\epsilon$ and consider all machines with the same height as a group. Denote these groups with $M_1,..,M_g$. We know that jobs from non-priority bags can go on any machine and we do not need to watch out for bag-constraints. Therefore we consider the following generalization of bag-LPT to distribute jobs to groups of machines.

For every non-priority bag with small jobs $B_l$ sort all jobs decreasing by their height and sort all groups of machine increasing by their average load. Assume without loss of generality  that $M_1,..,M_g$ is the resulting sorted list. Then for every $i\le g$ assign the first $|M_i|$ jobs of the sorted bag $B_l$ to $M_i$ and remove the jobs from the list.

We call this algorithm group-bag-LPT. We will now show that this algorithm will give us an assignment of jobs sufficient enough to construct a good schedule.

\begin{lem}\label{lem:grplpt}
Consider a given MILP solution and a schedule where large and medium jobs are already scheduled with respect to the MILP solution and small jobs of priority bags are scheduled evenly distributed in the reserved space on top of the patterns. Let $B_1,...,B_r$ be the non-priority bags with small jobs. Group-bag-LPT will schedule all small jobs from the non-priority bags such that the area assigned to any machine group $M_i$ will be bound by $|M_i|*(1+O(\epsilon)))$. Further applying bag-LPT to each machine group $M_i$ with the assigned jobs will yield a $1+O(\epsilon)$ schedule.
\end{lem}

This concludes the placement of non-priority bags. We can also observe that bag-LPT and group-bag-LPT both run in polynomial time. Now the last thing that is left to do is fill up the reserved space with the respective priority bags.

\subsection{Priority bags}

For priority bags we want to follow the same approach: Identify fitting groups of machines and schedule jobs such that the overall area is limited. To also respect the bag-constraints we will consider the machines holding the same pattern as per MILP solution as groups. As we reserved space for these jobs, we may further assume that each of these groups have the same height given by the height of the pattern.
Remark that for now we will talk about patterns as used in the MILP and we ignore for now whether a pattern was changed while using the swapping argument of lemma \ref{lem:large}. This might cause conflicts with the next scheduling step but we will show later that these can be resolved. To avoid having an arbitrary amount of fractional jobs per group of machines we merge together fractional jobs.
 
Given pattern $p$ with $m_p> 0$ machines and a bag $B_l$ such that the MILP assigns jobs from $B_l$ to $p$. Assume without loss of generality that $j_1,...j_{n_p} \in B_l$ are the small jobs that are assigned to $p$ and let $\alpha_i \in (0,1]$ for $i\le n_p$ be the fractional amount of $j_i$ that was assigned to $p$. Further assume that $n_f\le n_p$ is the amount of jobs that was assigned fractionally to $p$ and without loss of generality assume also that $j_1,...j_{n_f}$ are the jobs that were assigned fractionally. We can see that $\alpha_i < 1$ for $i\le n_f$ and $\alpha_i = 1$ for $n_f < i \le n_p$. Set $m_f := m_p -(n_p-n_f)$ to be the number of machines that we need to distribute all fractional jobs on.  We modify $B_l$ by removing all jobs $j_i$ with $i \le n_f$ and replacing them with $m_f$ jobs that all have the same height given by $h_f := \sum\limits_{i \le n_f}{\frac{h_{j_i} *\alpha_i}{m_f}}$. Basically these new jobs are made out of equal sized pieces of each fractional job, so we gain $m_f$ new jobs, each consisting of a $\frac{1}{m_f}$ part of each fractional assigned job. All fully distributed jobs will stay the same. 

Now we have a situation where we have as many small jobs from the pattern as we have machines and therefore we can use bag-LPT to distribute them.

\begin{coro}\label{cor:prio}
Given an MILP solution and a pattern $p$. Let $A_p$ be the average area assigned to each machine of $p$ by small jobs in the MILP. By using bag-LPT on the $m_p$ assigned machines and the modified small jobs from priority bags as described above, we gain a distribution of modified small jobs that increases the height of each machine by at most $A_p + \epsilon^{k+1}$.
\end{coro}

Now given this schedule, that contains modified jobs, we can just keep the positions for fully distributed jobs. To insert the original fractionally distributed jobs, we use the new constructed jobs as slots to fill them in. By rounding up the height of these slots, we can ensure the height increase is bound.

\begin{lem}\label{lem:assign}
Given an MILP solution and a placement of large and medium jobs based on the MILP solution. Let there be an assignment of modified small jobs from priority bags based on corollary
\ref{cor:prio}. Then there is an assignment of all small jobs from priority bags to machines such that the load increases by at most $O(\epsilon)$.
\end{lem}

Now we have constructed a complete schedule that contains all jobs and so far has a height of $1+O(\epsilon)$. We also know that non-priority bags are conflict free, as per our algorithm. For priority bags however we cannot guarantee feasibility in regards to the bag-constraints yet.
Since we moved large jobs during our proof of lemma \ref{lem:large} and we ignored this potential change of patterns while distributing small jobs, our schedule so far could have conflicts. To resolve these conflicts we do the same as for lemma \ref{lem:large} and move jobs around.

\begin{lem}\label{lem:conf}
Given a schedule $S$ constructed so far including all jobs. If $S$ after placing small jobs from priority bags contains conflicting jobs, then we can resolve these conflicts in polynomial time and will increase the height of the schedule by at most $O(\epsilon)$.
\end{lem}

With this last lemma we can resolve the last conflicts and gain a feasible $1+O(\epsilon)$ schedule for our modified problem instance. With the methods described in section 2 we can build therefore a respective solution for our original instance. The running time is dominated by finding the solution for the MILP. Since all other scheduling steps run in polynomial time we end up with an overall running time of the form $poly(|I|) * f(\frac{1}{\epsilon})$ for our algorithm. Finally this concludes the proof of our theorem and the main result of this paper.

\section{Conclusion}

In this paper we have proven that there exists an EPTAS for machine scheduling with bag-constraints. With this we actually solved one of the open questions mentioned in \cite{bagptas}. We have done so by introducing a convenient instance transformation and a MILP that allows us to schedule all large jobs of the instance using a constant number of integral variables. Overall we showed that is sufficient to place a constant number of large jobs to find a schedule for all large jobs. Further we concluded with the given transformation that we can ignore some bag-constraints between large, medium and small jobs of some bags. 

Since the problem is strongly NP-hard it is unlikely that an FPTAS exists, unless $P=NP$. Nevertheless bag-constraints come with other open problems. For example one can consider another machine model. There has been some research on unrelated machines \cite{bagptas} but other machine models have not been looked at yet. Furthermore one can also consider other optimization functions for these types of constraints. It would also be interesting to see whether the techniques used in this paper still hold in any of these possible variations or whether they can be of use for other areas and problems with similar constraints.


\bibliography{References}

\newpage

\appendix

\section{Proofs of section \ref{S:Prel}}

\begin{proof}[Proof of lemma \ref{lem:mod1}]
Let $S$ be the solution for $I$ with makespan $C$. We will construct a solution $S'$ for $I'$ based on $S$ as follows: First of all schedule all jobs of priority bags the same way as in $S$. For non-priority bags we will do the same. More precise, schedule all original small jobs and large jobs from non-priority bags on the same machines as in the solution $S$. Every filler-job $\overline j$, that was added in $I'$, corresponds to a medium or large job $j$ in $S$, so place $\overline j$ on the same machine as $j$ was placed on $S$.

All jobs are now placed in the schedule and it is left to show that the schedule is feasible and bound in height. We will prove the first by indirect proof, so assume that $S'$ is not feasible and two jobs $j_1,j_2$ of the same bag are on the same machine $i$. Due to the way we placed $j_1,j_2$ on these machines we can follow that there must be $j'_1,j'_2$ in $S$ that were placed on the same machine while $j'_1,j'_2 \in B_l$ for one non-priority bag $B_l$. This means $S$ is an invalid solution, which is a contradiction. 

Now since $S'$ must be feasible we look at the makespan. Due to our construction, the solution resembles $S$ except for the filler-jobs of large jobs that we put on top of each machine. To be more precise, on each machine $i$ the load due to large jobs is the same. The load of medium jobs is either the same or even smaller when a medium job of non-priority bags got removed and replaced by a small filler-job. And finally the load due to small jobs is the same except that we added filler-jobs that correspond to large jobs in the original instance. So for all large jobs on a machine that machine gets an additional load of a small job.
Note that the number of large jobs any machine can hold is bound by $\frac{C}{\epsilon^k}$ since otherwise the makespan could not be $C$. Overall we get that the maximum height of any machine in $S'$ is bound by $C +\frac{C}{\epsilon^k} *\epsilon^{k+1} = C+\epsilon *C = (1+\epsilon)*C$.
\end{proof}

\begin{center}
\begin{figure}[bh]
  \centering
  
	\includegraphics[scale=0.4]{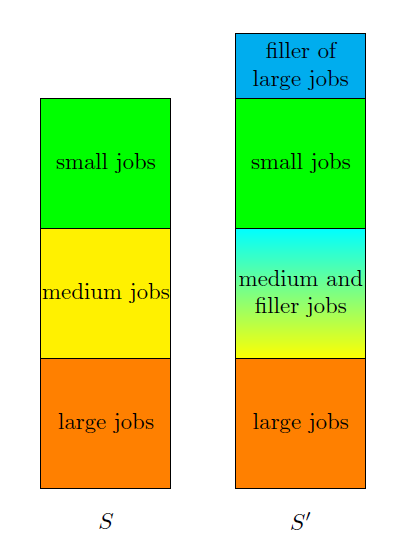}
	  \caption{}
		{\footnotesize Comparison of one machine in the solutions generated in the proof of lemma \ref{lem:mod1}. Left machine belongs to $S$ and right machine to modified solution $S'$. The machines hold the same large and small jobs (not accounting for filler-jobs). Medium jobs from non-priority bags are equal as well. Differences are that in $S'$ we have filler-jobs for medium jobs instead of the respective medium jobs in $S$ and additionally filler-jobs that correspond to large jobs.}
  \label{abb:modsol}
	
\end{figure}
\end{center}

\begin{proof}[Proof of lemma \ref{lem:mod2}]
Let $B_l$ be a non-priority bag of instance $I$ and denote with $B_l^{med}$ the set of medium jobs of $B_l$. Let $B'_l$ be the bag of exclusively $I'$ that contains all large jobs of $B_l$. Our goal is now to add jobs of $B_l^{med}$ to solution $S'$ such that no more than two jobs of $B_l^{med} \cup B'_l$ are on the same machine. Let $M^{B'_l}$ be the set of machines that does not hold any job of ${B'_l}$. We can follow that $|B_l^{med}| \le |M^{B'_l}|$ as $B_l^{med}$ and $B'_l$ made up one bag in our initial instance. We start our distribution of medium jobs with a fractional distribution that assigns jobs evenly on all free machines.  We do this by defining a vector $(x_{i,j})_{i\in M, j\in J_{non-prio}^{med}}$ with $J_{non-prio}^{med}$ being the set of all medium jobs of non-priority bags and 
$x_{i,j} := 
\begin{cases}
       \frac{1}{|M^{B'_l}|} &\quad\text{if } \exists l:  j\in B^{med}_l \text{ and } i\in  M^{B'_l}  \\
       0 &\quad\text{otherwise.} \\ 
\end{cases}$

Basically this construction will schedule a part of a medium job on all machines that are free for it, in the sense that it does not contain a job of $B'_l$. We get that for each machine $i$ and a bag $B_l$ it holds that $\sum\limits_{j\in B^{med}_l}{x_{i,j}} = \frac{|B^{med}_l|}{|M^{B'_l}|}$ if $i \in M^{B'_l}$, otherwise this sum will equal to $0$. Furthermore we know that $B_l$ in $I$ was a small bag (which implies that $B'_l$ is a small bag as well) and per definition we know that $|B'_l| \le \epsilon *m$, therefore 
$|M^{B'_l}| \ge (1-\epsilon)m$. Denote with $A$ the set of indices of all priority bags, so we have that $B_l$ with $l\le b$ and $l \not \in A$ is a non- priority bag. With this we can conclude that the number of jobs that any machine $i$ can receive is bound by:

\begin{align*}
\sum\limits_{l\le b: l\not \in A} \sum\limits_{j\in B_l^{med}} x_{i,j} &\le \sum\limits_{l\le b: l\not \in A} \frac{|B^{med}_l|}{|M^{B'_l}|} &\\
&\le \sum\limits_{l\le b: l\not \in A} \frac{1}{(1-\epsilon) m} * |B^{med}_l| &\\
&\le \frac{1}{(1-\epsilon) m} * \frac{\epsilon^2m}{\epsilon^{k+1}} & (*)\\
&\le \frac{1}{(1-\epsilon)*\epsilon^{k-1}}&
\end{align*}
The inequality for $(*)$ holds as the area of all medium jobs is bound by $\epsilon^2*m$ by our definition of $k$ via lemma \ref{lem:medk} and therefore we get that the number of medium jobs is also bound by $\frac{\epsilon^2m}{\epsilon^{k+1}}$.

We can now construct a directed flow network as follows: We have nodes $v_l$ and $w_i$ for each non-priority bag $B_l$ and each machine $i$ and a source nodes $s$ and sink $t$. Further our network is given by edges of the type $(s,v_l)$ with capacity $|B_l^{med}|$ for any $v_l$ and $(w_i,t)$ with a capacity of $\left\lceil \sum\limits_{j}x_{i,j}  \right\rceil$, which equals to the number of jobs that was assigned to $i$ in our constructed vector. Finally we have edges $(v_l,w_i)$ of capacity $1$ if and only if $i \in M^{B_l'}$. 

By intuition this network simulates an assignment of jobs to machines, as each path from $s$ to $t$ resembles an assignment of a job to a machine. Our constructed vector $x$ resembles a fractional assignment and we can follow there is a fractional flow for this network with value $\sum\limits_{l\le b: l\not \in A} |B_l^{med}|$. Flow theory implies that there exists an integral solution. As edges from bags to machines have capacity $1$ at most one medium job of any bag can be assigned to each machine and further this machine does not hold any large job of the respective bag with large jobs.
Finally we can conclude that, due to the capacity of the sink edges, we know this assignment will assign at most $\frac{1}{(1-\epsilon)*\epsilon^{k-1}}+1 \le \frac{2}{\epsilon^{k-1}}$ jobs to any machine. Overall we can compute this assignment in polynomial time and our solution increases in height by at most $\frac{2}{\epsilon^{k-1}} *\epsilon^{k} = 2\epsilon$ which concludes the proof.
\end{proof}

\begin{proof}[Proof of lemma \ref{lem:mod3}]
First of all we will take a loot at our current situation and compare $I$ and $I'$. Looking at the jobs, $I'$ holds the same jobs as $I$ and additional filler-jobs for each large and medium job of non-priority bags in $I$. For each non-priority bag that appears in $I$ we have two bags $B_l$ and $B'_l$ in $I'$ and the jobs of the original bag are separated among these two bags, with $B_l$ holding small jobs and the additional filler-jobs and $B'_l $ containing all large and medium jobs. Consider a solution $S'$ for $I'$. We will generate a solution $S$ by merging each bag pair together again, while using filler-jobs to remove conflicts that occur after merging these bags.

Consider therefore a third instance $I''$ that is given as a result of merging all bag pairs $B_l \cup B'_l$ of the instance $I'$ and keeping priority bags as they are. Basically $I''$ has the same set of jobs as $I'$ and the same bags as $I$ (just that non-priority bags were extended with filler-jobs). Consider the schedule $S'$ as a solution for $I''$. In case that $S'$ is valid for $I''$ regarding the bag-constraints we can remove all filler-jobs and we have a solution for $I$ that is feasible and has the same or lower makespan. In general this will not work and we will most likely end up with conflicting jobs on some machines. In fact given our definition it is not even guaranteed that a valid schedule for $I''$ exists as we might have that $|B_l \cup B'_l|>m$. We can neglect this problem though, as we only want to use $I''$ to show, that we can fix conflicts and remove filler-jobs to gain a solution for $I$. We will do both things in one step by swapping jobs around such that only filler-jobs cause conflicts.

For easier notation we will notate bags from $I''$ with $\overline{B_l}$, so we have that $\overline{B_l} = B_l \cup B'_l $ for the respective non-priority bags $B_l,B'_l$ of $I'$.
So consider a non-priority bag $\overline{B_l}$ that causes conflicts in $S'$ when considered as a solution for $I''$. Let $g$ be the number of medium and large jobs of $\overline{B_l}$, $f$ the number of filler-jobs in $\overline{B_l}$ and $c$ the number of conflicts. We know by definition of our modification that $f =g$. Since we know that $S'$ was a valid solution for $I'$ there cannot be conflicts between jobs of the same bag $B_l$ or $B'_l$. For that reason conflicts must occur between a small and a large/medium job of $\overline{B_l}$. Together we get that $f=g \ge c$ as we can have only as many conflicts as we have large and medium jobs. 

This means that for every conflict we have a filler-job somewhere in our schedule. The idea now is simple. If there is a conflict between two jobs and the small job is a filler-job we do nothing. If the conflicting small job, call it $j$, is not a filler-job we know there must be a filler-job $\overline j$ on a non-conflicting machine. We now can swap the position of these two jobs. By definition we know that the load of the non-conflicting machine does not increase as $p_j \le p_{\overline j}$. Note that the height of the other machine might increase, but as we are about to remove filler-jobs, so we can ignore this. By doing this kind of swap for every non-filler-job that causes a conflict we can construct a situation where our solution only has conflicts that are caused by filler-jobs. We can now remove all filler-jobs and end up with a solution $S$ that has no conflicts. Furthermore this solution is a feasible solution for the instance $I$ and has the same makespan as the solution $S'$ or even a lower makespan.
\end{proof}

\section{Proofs of section \ref{S:MILP}}

\begin{proof}[Proof of lemma \ref{lem:milp}]
Let $S$ be the solution with makespan $T$. Initiate all variables of the MILP solution with $0$. For each machine $i$ identify the pattern $p$ that was scheduled on $i$ in the solution $S$ and increment $x_p$. For all small jobs $j$ scheduled on $i$ let $B_l^s$ be the size-restricted bag that contains $i$ and increment $y_p^{B_l^s}$. We show that this generated solution must be valid. We can see that constraints (1), (2) and (3) must be satisfied. If this were not the case this would mean that there are either more than $m$ machines or not all jobs were scheduled, which contradicts $S$ being a valid solution. Also constraints (6), (7), (8) and (9) must be satisfied as we constructed an integer solution.

Assume that constraint (4) is not satisfied, let $p$ be the pattern of the failed constraint. We know that $x_p = 0$ implies that $y^{B_l^s}_p = 0$ for any $B_l^s$ based on how we set up our variables. So we get that constraint (4) must have failed, because the total load that was assigned to machines with pattern $p$ was too much and exceeded $T*x_p$. Since we got this assignment from our solution $S$ we can conclude that the total load on these $x_p$ machines exceeds $T*x_p$ making it impossible for $S$ to have a makespan of $T$.

Finally assume now that constraint (5) failed and let $p$ be the pattern and $B_l$ be the bag of the failed constraint. As in the previous case we can follow that $x_p \neq 0$. At least $x_p +1$ jobs of bag $B_l$ were assigned to $x_p$ machines by considering two cases dependant on $\chi_p(B_l)$. If $\chi_p(B_l) = 0$, then the constraint must have failed due to more than $x_p$ small jobs being assigned to $p$ and if $\chi_p(B_l) =1$ the constraint failed as at least one job of $B_l$
was assigned to $p$ and $p$ also  contains large jobs of $B_l$. This means that the solution $S$ must have placed $x_p+1$ jobs of the same bag on $x_p$ machines, making solution $S$ infeasible.

Overall we can conclude that all constraints must be satisfied by our solution.
\end{proof}

\begin{proof}[Proof of lemma \ref{lem:runtime}]
We will start with the number of priority bags to calculate the number of patterns.
By definition \ref{def:prio} we know that for each large item size we have at most $b'$ priority bags. We can bound the number of priority bags, given by $|A|$, with $d*b' = d*q*(q*d+1) \in O(d^2*q^2)$. We know that $d\in O(log_{1+\epsilon}(\frac{1}{\epsilon^k}))$ and we can conclude with
\begin{align*}
log_{1+\epsilon}(\frac{1}{\epsilon^k}) &= \frac{log(\frac{1}{\epsilon^k})}{log(\epsilon+1)} \\
&\le  \frac{log(\frac{1}{\epsilon^k})}{\frac{\epsilon}{1+\epsilon}} \\
&=   k*log(\frac{1}{\epsilon})*\frac{1+\epsilon}{\epsilon} \\
&\le   \frac{1}{\epsilon^2} *log(\frac{1}{\epsilon})*(\frac{1}{\epsilon} +1)
\end{align*}  
that $d\in O(\frac{1}{\epsilon^3} *log(\frac{1}{\epsilon}))$. By definition we also know that $q\in O(\frac{1}{\epsilon^{k+1}})$. Together we get that the number of priority bags $|A| \in O(\frac{1}{\epsilon^{2k+2+6}} *log^2(\frac{1}{\epsilon}))$, this even holds when adding the number of large bags as this number is bound by $O(\frac{1}{\epsilon^{k+2}})$. 

We know by definition \ref{def:pat} that the number of possible entries in a pattern is bound by $d_m*(|A|+1) \in  O(\frac{1}{\epsilon^{2k+8+3}} *log^3(\frac{1}{\epsilon}))$, with $d_m$ being the number of medium item sizes. With the same argumentation as for $d$ we know that $d_m\in O(\frac{1}{\epsilon^3} *log(\frac{1}{\epsilon}))$.
We can further conclude that the number of all patterns can be bound by $(d_m* (|A|+1))^q$. Now we include the variables of constraint (7). With the same argument again we can see that the number of small item sizes in (7) is bound by $O(\frac{1}{\epsilon^3} *log(\frac{1}{\epsilon}))$ and together we can finally bound the number of used integer variables by  
$O((\frac{1}{\epsilon^{2k+11}}*log^3(\frac{1}{\epsilon}))^{q+1}) = 2^{O(\frac{1}{\epsilon^{k+1}}*log(\frac{1}{\epsilon^{2k+11}}*log^3(\frac{1}{\epsilon})))  }$.

With this we can give the total running time with the result of Kannan, which depends on the number of integral variables $z = 2^{O(\frac{1}{\epsilon^{k+1}}*log(\frac{1}{\epsilon^{2k+11}}*log^3(\frac{1}{\epsilon})))}$. We achieve a running time of $z^{O(z)}*poly(g) = 2^{O(z*log(z))}*poly(g)$ with $g$ being the length of the input. As the number of constraints (excluding variables constraints) is bound by $O(n^2)$ we get the desired running time.
 
\end{proof}

\begin{proof}[Proof of lemma \ref{lem:large}]

For starters we place all large and medium jobs of priority bags as given in the MILP solution. This placement will not only be feasible so far but also contain every medium job of our instance. Now for the non-priority bags we first of all will assume that all large item sizes $s_1 ,..,s_d$ are sorted such that $|B_{o_{s_i}(b')}^{s_i}|\le |B_{o_{s_{i+1}}(b')}^{s_{i+1}}| $ for all $i<d$ (remember that $o_s$ is the function giving the ordered indices based on the priority bag definition \ref{def:prio}). We start by placing all items of size $s_1$, then $s_2$ and so on, so we can assume for item size $s_i$ with $i\le d$ that all items of sizes $s_{i'}$ with $i'<i$ are already placed with respect to the bag-constraint. We remember further the definition of $b' = q*(dq+1)$ and set $z := (dq+1)$ for the following. 

Consider a slot of size $s_i$ reserved for a job of bag $B_x$  by the MILP. Recall that $B_x^{s_i} \subseteq \bigcup\limits_{b'< l \le b}{B_{o_{s_i}(l)}^{s_i}}$ and choose a non-priority bag $B^{s_i}_{o_{s_i}(l)}$ with $b'<l\le b$ and $l\not \in A$ such that $B^{s_i}_{o_{s_i}(l)}$ has the maximum number of jobs and does not violate against any bag constraint on the machine the slot is on. Place any job of $B^{s_i}_{o_{s_i}(l)}$ on the respective machine, remove the assigned job from the bag and continue with the next slot. In fact any greedy based algorithm to distribute jobs to slots will suffice for this, as we cannot guarantee to not run into a conflict.

Such a conflict may arise when we are forced to place a job on a machine, that already holds another job of the same bag. This scenario can be unavoidable depending on the MILP solution. Assume without loss of generality that we want to place a job $p \in B^{s_i}_{r}$, with $B_{r}$ being a non-priority bag, on a machine $c$ that already has a job of $B_{r}$ assigned. We want to solve this conflict by finding another job $p'$ of size $s_i$ that was placed on another machine $d$. We then can swap the slots of $p$ and $p'$ and schedule $p$ on the machine $d$, while putting $p'$ on $c$. We will see that we can find $p'$ such that there will be no conflicts on neither $c$ nor $d$.

Therefore consider the bags $B^{s_i}_{o_{s_i}(l)}$ for $l\le z$. Due to the ordering, given by the index permutation $o_{s_i}$, we can conclude that these bags all together hold at least 
\begin{equation}
\sum\limits_{l\le z}{|B^{s_i}_{o_{s_i}(l)}|} \ge z* |B^{s_i}_{o_{s_i}(z)}| \ge z*  |B^{s_i}_{o_{s_i}(b')}|  \tag{*}
\end{equation} jobs. 
Since each machine can have at most $q$ medium or large items and since
\begin{align*}
\frac{z}{q} *  |B^{s_i}_{o_{s_i}(b')}|&=  \frac{(dq+1)}{q} *  |B^{s_i}_{o_{s_i}(b')}|  \\
&= (d+\frac{1}{q}) |B^{s_i}_{o_{s_i}(b')}| \\
&> d *|B^{s_i}_{o_{s_i}(b')}| 
\end{align*}
we can conclude that the items from the above fixated bags must be distributed among at least $d*|B^{s_i}_{o_{s_i}(b')}| +1$ different machines.
Furthermore we get that:
\begin{align*}
d*|B^{s_i}_{o_{s_i}(b')}| +1 &> d*|B^{s_i}_{o_{s_i}(b')}| & \\
&\ge i*|B^{s_i}_{o_{s_i}(b')}|  & i \le d\\
&\ge \sum\limits_{l \le i}|B^{s_i}_{o_{s_i}(b')}| &  \\
&\ge \sum\limits_{l \le i}|B^{s_l}_{o_{s_l}(b')}| & \text{order of item sizes} \\
&\ge \sum\limits_{l \le i}|B^{s_l}_{r}| & \text{order of bags}
\end{align*}

The last two inequalities hold due to the respective ordering for size-restricted bags and item sizes. The last inequality also uses the fact that $B_r$ is a non-priority bag and as such has a higher position in the sorted list $o_s$ than $b'$ for any item size. If this were not the case, then $B_r$  would be a priority bag to begin with. Therefore with the whole inequality we see there is strictly more machines that hold jobs from $B^{s_i}_{o_{s_i}(l)}$ for $l\le z$ than machines with jobs from $B_r$ so far. So choose $p'$ to be an arbitrary job from $B^{s_i}_{o_{s_i}(l)}$ for some $l\le z$ such that $p'$ currently is assigned to a machine $d$ not holding a job from $B_r$. If $p'$ causes no conflict on $c$ we can assign $p'$ to $c$ and $p$ on $d$, thus the conflict is solved. In case $p'$ causes a conflict, we then can consider another group of bags to find an alternative job to swap with $p$. Consider another group of bags $B^{s_i}_{o_{s_i}(l)}$ with  $(g)z< l\le (g+1)z$ for some $g< q$. The argumentation from above  still holds as the inequality $(*)$ is still true when considering these other groups of bags. Since there is at most $q -2$ other jobs on $c$ that may conflict with the replacement job $p'$ and we have $q$ groups to consider, we are sure to find one job to fix the conflict.

So with this approach we are able to schedule all large jobs without altering the objective of the MILP solution as we only swap around jobs of the same size.
\end{proof}

\section{Proofs of section \ref{S:Small}}

\begin{proof}[Proof of lemma \ref{lem:baglpt}]
The second part follows immediately from the first. In the final schedule the smallest machine can have a load of at most $h+x$ because otherwise the distributed area would exceed the total area that is given by $h*m + A = (h+x)*m$. Assuming the first part holds we get that the highest machine cannot be higher than $h+x+p_{max}$. We now proof the first part by induction over the number of bags $b$. 

The induction base for $b=1$ follows immediately by the fact that all machines have the same height $h$. As each machine receives one job, possibly of height 0 for dummy jobs, we get that the height difference of two machines is bound by $p_{max}$. As induction hypothesis assume that $b\ge 1$ bags are placed so far such that the height difference of any two bags is bound by $p_{max}$. For the induction step consider now the bag $B_{b+1}$. Consider further two machines $m_1,m_2$ with heights $h_1,h_2$ prior to placing $B_{b+1}$ and let $j_1,j_2 \in B_{b+1} $ be the jobs assigned to $m_1,m_2$ respectively as per bag-LPT and denote their heights with $p_1,p_2$.

Without loss of generality assume that $h_1 \le h_2$, or swap the machines if it does not hold. If $h_1 = h_2$ we have the same case as in the induction base and the claim follows immediately, since their different in height after placing $B_{b+1}$ can be at most $|p_1 - p_2| \le p_{max}$. So assume that $h_1 < h_2$ and by the property of bag-LPT we can follow that $p_2 \le p_1$. Let $d := h_2 - h_1$ be the difference of height before placing the bag $B_{b+1}$. Consider two cases given by the relationship of the new heights of machines.

Case 1: $h_1 + p_1 \le h_2 + p_2$, so $m_1$ still has the lower load of both machines. We get now for the new height difference $d_2$ that:
$d_2 = h_2 + p_2 - (h_1 + p_1)= h_2 - h_1 + p_2 - p_1 \le h_2 -h_1 = d$. So overall the difference of load is lower than before.

Case 2: $h_1 + p_1 > h_2 + p_2$. For this case as the order of machines switched, we get that the new load difference is bound by the difference of the height of jobs $ p_1 - p_2 \le p_{max}$ as $j_1$ has to cover the previous height difference and the height of $j_2$ to change the order of machines. 

Finally by our induction we can conclude that the statement holds.

\end{proof}

\begin{proof}[Proof of lemma \ref{lem:grplpt}]
Consider the groups of machines $M_1,..,M_g$ and note that $g \le \frac{1}{\epsilon} + O(1)$ since the height of our MILP solution is bound by $1+2*\epsilon + \epsilon^2 \le 1+3*\epsilon$. Without loss of generality let $M_1,..,M_g$ be sorted non-decreasing by their average heights after running the group-bag-LPT and denote with $h_1,..,h_g$ these heights. Let $1 < i\le g$ be the largest index such that $h_{i} - h_{i-1} > p_{max} $ with $p_{max}\le \epsilon^{k+1} $ being the height of the biggest small job in a non-priority bag. If there is no such index set $i:= 1$. Note that now for $i< i' \le g$ we have that $h_{i'} - h_{i'-1} \le p_{max}$.

Let $L:= \bigcup\limits_{i\le i' \le g}{|M_{i'}|}$ be the set machines in $M_i, .., M_g$. Since the difference in height between any machine $L$ and any machine not in $L$ is larger than $p_{max}$ and therefore by definition of group-bag-LPT all machines in $L$ must have received the $|L|$ smallest jobs of each bag. If this were not the case then both machines would have a load difference $\le p_{max}$. Further we know the MILP distributed jobs from small bags, such that the total area assigned to machines in $L$ is bound by $L(1+3\epsilon)$. As machines in $L$ receive only the smallest jobs we can also conclude that the total area assigned to $L$ after group-bag-LPT is the same or even smaller than in the MILP and thus is also bound by $L(1+3\epsilon)$. With this we can also see that the average height of the machine group $M_i$ is bound with: $h_i \le 1 + 3 \epsilon$.
Overall we can now conclude the average height of the largest machine group. We have that:
$h_g \le (g-i)* p_{max} + h_i \le (g-1) * \epsilon^{k+1} + 1 +3\epsilon = 1 + O(\epsilon)$

Finally with lemma \ref{lem:baglpt} we can conclude that after applying bag-LPT to all groups of machines their respective height is bound by 
$1 + O(\epsilon)$, since we only have small jobs to distribute.
\end{proof}

\begin{proof}[Proof of corollary \ref{cor:prio}]
The proof follows from lemma \ref{lem:baglpt}. We remark that with  constraint (5) we ensured that at most $x_p$ jobs could be assigned to each pattern. Therefore we get that:

\begin{equation*}
\sum\limits_{i \le n_f}\alpha_i \le m_f.
\end{equation*}

Knowing this we can conclude that with $h_{max} := \max\limits_{i \le n_f}{h_{j_i}}$ that the height of any constructed job $h_f$ is also bound by:

\begin{equation*}
h_f = \sum\limits_{i \le n_f}{\frac{h_{j_i} *\alpha_i}{m_f}} \le \sum\limits_{i \le n_f}{\frac{h_{max} *\alpha_i}{m_f}} = 
\frac{h_{max}}{m_f} *\sum\limits_{i \le n_f}\alpha_i \le  h_{max} \le \epsilon^{2k+11}< \epsilon^{k+1}
\end{equation*}

\end{proof}

\begin{proof}[Proof of lemma \ref{lem:assign}]
We start by rounding the height of all fractional constructed jobs up to $\epsilon^{2k+11}$. With the same estimation of the proof of corollary
\ref{cor:prio} we get that the previous height of any these jobs was equal or smaller than $\epsilon^{2k+11}$. As we have $O(\frac{1}{\epsilon^{2k+10}})$ priority bags this rounding may increase the load on any machine by $O(\epsilon)$.

Now we see our rounded constructed jobs as slots for the actual small jobs that were fractionally distributed in the MILP. We know through our rounding that any small job that we need to distribute fits in any slot, so we show that we have enough slots to accommodate all jobs. Consider a bag $B_l$ with $n_l$ jobs left to distribute and let $s_l$ be the number of slots that consisted of fractional parts from jobs of $B_l$. Assume for an indirect proof that $s_l < n_l$. As we have only $s_l$ slots the MILP must have distributed all $n_l$ jobs among $s_l$ machines fractionally. This can only happen if too many jobs were assigned to one pattern, which is a contradiction to constraint $(5)$ of the MILP.

\end{proof}

\begin{proof}[Proof of lemma \ref{lem:conf}]
As by our algorithm we can conclude that any conflict may only arise through applying the techniques in lemma \ref{lem:large}, that is swapping a large job of a priority bag with a large job of a non-priority bag to resolve conflicts. Therefore we can conclude the only type of conflict arises between a large and a small job from priority bags. Note that a conflict between small and medium jobs is impossible as we never move medium jobs around (remember that by modification non-priority bags have no medium jobs) and a conflict between these kinds of jobs would contradict a feasible solution of the MILP.

When a small and a large job are conflicting, we want to use the machine the large job was initially placed on as the new machine for the small one. Therefore define for every priority bag $B_l$ and for every large job $j\in B_l $ $origin_l(j)$ to be the machine that $j$ was assigned to in the MILP solution. We can observe that this $origin_l$ function for every priority bag $B_l$ is injective, as every large job was assigned to exactly one unique machine. Further we can conclude that for a large job $j\in B_l$ and machine $i = origin_l(j)$ that in our current solution $i$ cannot hold a small or a medium job from $B_l$, as this would contradict either constraint $(5)$ of the MILP (for a small job) or the definition of patterns (for a medium job). Machine $i$ may however still hold either $j$ or another large job, that was moved there after $j$ was moved away. With these observations we consider the following strategy of removing conflicts:

Let $j_{small}, j_{large}$  be a pair of conflicting jobs on a machine $i_c$ from a bag $B_l$. Let $i = origin_l(j_{large})$ and consider $i$ as a new machine for $j_{small}$. In case $i$ is free, as in it does not hold any job from $B_l$ we are done. In the other case we have a job $j$ on machine $i$ and we can conclude with our observations that $j$ must be large. So set $i = origin_l(j) $ and consider this machine now as a potential new machine for $j_{large}$. We continue this until we find a free machine.

We will prove the correctness of this procedure by showing two things: First we prove that this procedure terminates and eventually finds a free machine. Secondly we will show that this machine is unique and no two jobs from the same bag will end up on the same machine. Let therefore $j_{small}, j_{large}\in B_l$  be a pair of conflicting jobs on a machine $i_c$. First off note that the $origin_l$ function will never point to $i_c$, since $i_c$was assigned a small job by the MILP and a large job would imply a violation of constraint $(5)$. Further we have that $origin_l$ is injective, so every new large job from $B_l$ that potentially blocks a machine, will point to a new machine and eventually one has to be free.

For the same reason it is also impossible that two conflicting small jobs end up on the same machine. Consider additionally to the previous situation a job $j_{small}^2$ on a machine $i_c^2$ causing a conflict. First since we distribute all small jobs to different machines we get that $i_c^2 \neq i_c$.  We can further conclude that $i_c^2$ will never be visited while finding a spot for $j_{small}$, since this would violate the MILP again in constraint (5). Also any machine/large job seen while trying to find a spot for $j_{small}^2$ will never point to a machine visited by $j_{small}$, cause this would violate the injectivity of the $origin_l$ function. So we get that our repair strategy terminates and also finds a feasible schedule.

Applying one repair step potentially increases the height of our schedule. The height of a machine will potentially increase when the MILP assigned a large job to a machine $m$ and we move this job away. So for every large job we move away from the machine it was assigned to by the MILP the height of this machine might increase. As each machine holds at most $\frac{1+2\epsilon}{\epsilon^{k}}$ large jobs we can bound the height increase by $\frac{1+2\epsilon}{\epsilon^{k}} + \epsilon^{k+1} = \epsilon + 2\epsilon^2$.

\end{proof}

\end{document}